\documentclass[12pt,a4paper]{article}
\usepackage{fullpage}
\usepackage{amsfonts,amsmath,amssymb}
\usepackage{amsthm}
\usepackage{graphicx}
\usepackage{sectsty}
\usepackage{authblk}

\theoremstyle{plain}
\newtheorem{theorem}{Theorem}[section]
\newtheorem{lemma}[theorem]{Lemma}
\newtheorem{proposition}[theorem]{Proposition}
\newtheorem{corollary}[theorem]{Corollary}

\theoremstyle{definition}
\newtheorem{definition}[theorem]{Definition}
\theoremstyle{remark}

\numberwithin{equation}{section}

\sectionfont{\large}
\newenvironment{acknowledgement}[1][Acknowledgement
]{\begin{trivlist} \item[\hskip \labelsep {\bfseries
#1}]}{\end{trivlist}}

\begin{document}
\title{Plancherel-P\'{o}lya's Type of Instability \\in Vibration System with Multiple Frozen Arguments }
\author[1]{ Lung-Hui Chen}
\author[2]{Chung-Tsun Shieh}

\affil[1]{\footnotesize General Education Center, Ming Chi University of Technology, New Taipei City, 24301, Taiwan;\newline Email: mr.lunghuichen@gmail.com.}
\affil[2]{\footnotesize Department of Mathematics, Tamkang University, New Taipei City, 25137, Taiwan;\newline Email: ctshieh@mail.tku.edu.tw.}
\maketitle
\begin{abstract}
We discuss the problem of the inverse spectral problem of Sturm-Liouville operator 
with multiple frozen arguments at $\{a_{1}, a_{2},\ldots,a_{N}\}$ in $(0,\pi)$. One may consider the characteristic functions as  perturbation of sine or of cosine functions depending on the boundary problem prescribed. However, such perturbation is represented in the form of Fourier transform of certain function which may or may not bring in Riesz basis theory and classical perturbation theory in functional analysis. We shall demonstrate the spectral perturbation in Plancherel-P\'{o}lya's type of inequality and connect to perturbation of related potential functions in $L^{2}$-functional norm. 
\end{abstract}
\section{Introduction}
In this paper, we study the instability problem on recovering the real-valued potential function $q(x)\in L^{2}(0,\pi)$ from the spectrum of boundary value problem
$$L=L(q(x),a_{1}, a_{2},\ldots,a_{N},\alpha,\beta)$$ 
in form of
\begin{eqnarray}
\left\{%
\begin{array}{ll}\label{1.1}
ly:=-y''(x)+q(x)\sum_{i=1}^{N} y(a_{i})=\lambda y(x),\,0<x<\pi;\vspace{10pt}\\
y^{(\alpha)}(0)=y^{(\beta)}(\pi)=0,
\end{array}%
\right.
\end{eqnarray}
in which $\lambda=\rho^{2},\,\rho\in\mathbb{C},$ is the spectral parameter, and $\alpha$, $\beta\in\{0,1\}$. In system~(\ref{1.1}), we assume that $0<a_{1}<a_{2}<\cdots<a_{N}<\pi$. The operator $l$ is called the Sturm-Liouville type of operator with $N$ frozen arguments that is compatible with either Dirichlet or Neumann boundary condition on the end points of $[0,\pi]$.
\par
In mathematical or engineering literature, the system~(\ref{1.1}) is called one kind of loaded differential equations. The system simulates various experiments in applied physics, mathematical physics, and many engineering projects focusing on the control and measurement of certain vibrations, magnetic shielding in certain oscillatory systems \cite{Dikinov,Iskenderov,Lomov,Nakhushev}. In oscillatory systems, one often needs to place certain measurement instruments or sensors to monitor the on-site physical states of the vibration systems which justifies the physical senses of frozen points in~(\ref{1.1}). 
\par
The inverse spectral problems for Sturm-Liouville operators with frozen arguments were previously studied in \cite{Alb,Alb1,Bon,Buterin1,Buterin2,Buterin3,Buterin4,Buterin,D,K1,K2,K3,Shieh} to represent a limited selection, and is highly-pursued by many mathematicians and engineers. Among the rapidly growing researches as aforementioned, there are concerns on the local stability problem about how the spectral data of system~(\ref{1.1}) is related to the potential function $q(x)$ or the other classes of information provided by the characteristic functions. Let us refer to \cite{Buterin1,Buterin2,Buterin3,Buterin4,K2,K3} to list merely a few recent results. The problem seems to be ill-posed on many respects of the system~(\ref{1.1}), so we investigate the instability of problem~(\ref{1.1}) via the point of view of interpolation theory in Fourier analysis, and Plancherel-P\'{o}lya's type of inequality of the class of entire function of finite type plays a role.

\section{Preliminaries}
We begin by referencing several identities from \cite{Buterin1,Buterin2,Shieh}, specifically their identities (25), (26), (27), and (28) in Shieh and Tsai \cite{Shieh}. The eigenvalues of differential system~(\ref{1.1}) are the zeros of the characteristic functions referring to \cite{K1,Shieh}. Here are the characteristic functions all we need for system~(\ref{1.1}) for $\alpha,\beta\in\{0,1\}$:
\begin{eqnarray}\hspace{-6pt}
\Delta_N^{(0,0)}(\lambda) = \begin{vmatrix}\label{211}
\frac{\sin \rho a_1}{\rho} & \hspace{5pt}\int_0^{a_1} q(t) \frac{\sin \rho (a_1 - t)}{\rho} dt-1 & 1 & 1 & 1 & \dots & 1 & 1 \\
\frac{\sin \rho a_2}{\rho} & \int_0^{a_2} q(t) \frac{\sin \rho (a_2 - t)}{\rho} dt & -1 & 0 & 0 & \dots & 0 & 0 \\
\frac{\sin \rho a_3}{\rho} & \int_0^{a_3} q(t) \frac{\sin \rho (a_3 - t)}{\rho} dt & 0 & -1 & 0 & \dots & 0 & 0 \\
\vdots & \vdots & \vdots & \vdots & \vdots & \ddots & \vdots & \vdots \\
\frac{\sin \rho a_N}{\rho} & \int_0^{a_N} q(t) \frac{\sin \rho (a_N - t)}{\rho} dt & 0 & 0 & 0 & \dots & 0 & -1 \\
\frac{\sin \rho \pi}{\rho} & \int_0^{\pi} q(t) \frac{\sin \rho (\pi - t)}{\rho} dt & 0 & 0 & 0 & \dots & 0 & 0
\end{vmatrix},
\end{eqnarray}
\begin{eqnarray}\hspace{-5pt}
\Delta_N^{(0,1)}(\lambda) = \begin{vmatrix}
\frac{\sin \rho a_1}{\rho} & \hspace{3pt}\int_0^{a_1} q(t) \frac{\sin \rho (a_1 - t)}{\rho} dt - 1 & 1 & 1 & 1 & \dots & 1 & 1 \\
\frac{\sin \rho a_2}{\rho} & \int_0^{a_2} q(t) \frac{\sin \rho (a_2 - t)}{\rho} dt & -1 & 0 & 0 & \dots & 0 & 0 \\
\frac{\sin \rho a_3}{\rho} & \int_0^{a_3} q(t) \frac{\sin \rho (a_3 - t)}{\rho} dt & 0 & -1 & 0 & \dots & 0 & 0 \\
\vdots & \vdots & \vdots & \vdots & \vdots & \ddots & \vdots & \vdots \\
\frac{\sin \rho a_N}{\rho} & \int_0^{a_N} q(t) \frac{\sin \rho (a_N - t)}{\rho} dt & 0 & 0 & 0 & \dots & 0 & -1 \\
\cos \rho \pi & \int_0^{\pi} q(t)\cos \rho (\pi - t)dt & 0 & 0 & 0 & \dots & 0 & 0
\end{vmatrix},
\end{eqnarray}
\begin{eqnarray}\hspace{-4pt}
\Delta_N^{(1,0)}(\lambda) = \begin{vmatrix}
\cos \rho a_1& \hspace{-2pt}\int_0^{a_1} q(t) \frac{\sin \rho (a_1 - t)}{\rho} dt - 1 & 1 & 1 & 1 & \dots & 1 & 1 \\
\cos \rho a_2& \int_0^{a_2} q(t) \frac{\sin \rho (a_2 - t)}{\rho} dt & -1 & 0 & 0 & \dots & 0 & 0 \\
\cos \rho a_3 & \int_0^{a_3} q(t) \frac{\sin \rho (a_3 - t)}{\rho} dt & 0 & -1 & 0 & \dots & 0 & 0 \\
\vdots & \vdots & \vdots & \vdots & \vdots & \ddots & \vdots & \vdots \\
\cos \rho a_N& \int_0^{a_N} q(t) \frac{\sin \rho (a_N - t)}{\rho} dt & 0 & 0 & 0 & \dots & 0 & -1 \\
\cos \rho \pi& \int_0^{\pi} q(t) \frac{\sin \rho (\pi - t)}{\rho} dt & 0 & 0 & 0 & \dots & 0 & 0
\end{vmatrix},\label{224}
\end{eqnarray}
and
\begin{eqnarray}\hspace{-1pt}
\Delta_N^{(1,1)}(\lambda) = \hspace{-2pt}\begin{vmatrix}
\cos \rho a_1 &\hspace{-4pt} \int_0^{a_1} q(t) \frac{\sin \rho (a_1 - t)}{\rho} dt - \hspace{-2pt}1 & 1 & 1 & 1 & \hspace{-2pt}\dots & 1 & 1 \\
\cos \rho a_2 & \int_0^{a_2} q(t) \frac{\sin \rho (a_2 - t)}{\rho} dt & -1 & 0 & 0 & \hspace{-2pt}\dots & 0 & 0 \\
\cos \rho a_3 & \int_0^{a_3} q(t) \frac{\sin \rho (a_3 - t)}{\rho} dt & 0 & \hspace{-2pt}-1 & 0 &\hspace{-3pt} \dots & 0 & 0 \\
\vdots & \vdots & \vdots & \vdots & \vdots & \hspace{-3pt}\ddots & \vdots & \vdots \\
\cos \rho a_N & \int_0^{a_N} q(t) \frac{\sin \rho (a_N - t)}{\rho} dt & 0 & 0 & 0 & \hspace{-3pt}\dots & 0 & -1 \\
-\rho\sin\rho\pi & \hspace{-3pt}\int_0^{\pi} q(t)  \cos \rho (\pi - t) dt & 0 & 0 & 0 & \hspace{-3pt}\dots & 0 & 0
\end{vmatrix}.
\end{eqnarray}
Then, we deduce from~(\ref{211}), that is for the case $\alpha=\beta=0$,
the following integral identity
\begin{eqnarray}\nonumber
\Delta^{(0,0)}_{N}(\lambda)&\vspace{-2pt}=\vspace{-2pt}&\frac{\sin\rho\pi}{\rho}-\frac{\sin\rho\pi}{\rho^{2}}\Big\{\sum_{i=1}^{N}\int_{0}^{a_{i}}\sin{\rho(a_{i}-t)}q(t)dt\Big\}\\&&+\frac{\sum_{i=1}^{N}\sin\rho a_{i}}{\rho^{2}}\int_{0}^{\pi}\sin{\rho(\pi-t)}q(t)dt.
\end{eqnarray}
Since the characteristic function depends on the potential function $q$, we equivalently write
\begin{eqnarray}\nonumber
\Delta^{(0,0)}_{N}(q;\rho)&\vspace{-2pt}=\vspace{-2pt}&\frac{\sin\rho\pi}{\rho}-\frac{\sin\rho\pi}{\rho^{2}}\Big\{\sum_{i=1}^{N}\int_{0}^{a_{i}}\sin{(\rho t)}q(a_{i}-t)dt\Big\}\\&&+\frac{\sum_{i=1}^{N}\sin\rho a_{i}}{\rho^{2}}\int_{0}^{\pi}\sin{(\rho t)}q(\pi-t)dt,\label{662}
\end{eqnarray}
and the other cases have similar expressions. Therefore, one shall only demonstrate the proof for the case $(\alpha,\beta)=(0,0)$. 
\par
Furthermore, let $q^{1},\,q^{2}\in L^{2}(0,\pi)$ be potentials parametrizing the system~(\ref{1.1}) with the corresponding representation function 
\begin{eqnarray}
G^{j}(\rho)&:=&\Delta^{(0,0)}_{N}(q^{\,j};\rho),\,j=1,2,\\
\widehat{q}\,(x)&:=&q^{1}(x)-q^{2}(x),\nonumber
\end{eqnarray}
and
\begin{eqnarray}
\hspace{-2pt}\widehat{G}(\rho)&:=&\Delta^{(0,0)}_{N}(q^{1};\rho)-\Delta^{(0,0)}_{N}(q^{2};\rho)\\\nonumber
&=&-\frac{\sin\rho\pi}{\rho^{2}}\Big\{\sum_{i=1}^{N}\int_{0}^{a_{i}}\sin(\rho t)\,\widehat{q}(a_{i}-t)dt\Big\}\\&&+\frac{\sum_{i=1}^{N}\sin\rho a_{i}}{\rho^{2}}\int_{0}^{\pi}\sin(\rho t)\,\widehat{q}(\pi-t)dt.\label{2.8}
\end{eqnarray}
Let us denote the zero set of~(\ref{2.8}) as 
$$\mathcal{Z }(\widehat{G}):=\{\alpha_{1},\alpha_{2},\ldots\},$$ 
which includes the eigenvalue set of system~(\ref{1.1}) when $(\alpha,\beta)=(0,0)$. We also rephrase~(\ref{2.8}) to be
\begin{eqnarray}\nonumber
\widehat{G}(\rho)&=&
-\frac{\sin\rho\pi}{\rho^{2}}\Im\Big\{\int_{-\pi}^{\pi}e^{-i\rho t} \{\sum_{i=1}^{N}\,\chi_{[0,a_{i}]}(t)\,\widehat{q}(a_{i}-t)\}dt\Big\}\\&&+\frac{\sum_{i=1}^{N}\sin\rho a_{i}}{\rho^{2}}\Im\Big\{\int_{-\pi}^{\pi}e^{-i\rho t}\chi_{[0,\pi]}(t)\,\widehat{q}(\pi-t)dt\Big\},\label{2.88}
\end{eqnarray}
in which the cut-off function on $[a,b]$
\begin{eqnarray}
\chi_{[a,b]}(t):=
\left\{%
\begin{array}{ll}
1,\,\mbox{ for }t\in[a,b];\\
0,\,\mbox{ otherwise}.
\end{array}
\right.
\end{eqnarray}
Now, we apply the Paley-Wiener theorem \cite[p.\,151]{Levin2}, for each function $\psi(t)\in L^{2}(-\pi,\pi)$, there exists the representation
\begin{equation}
F(\rho)=\frac{1}{2\pi}\int_{-\pi}^{\pi}e^{-i\rho t}\psi(t)dt,\mbox{ where }F(\rho)\in L^{2}_{\pi}(-\infty,\infty),
\end{equation}
which denotes entire functions that is $ L^{2}$-integrable on the real axis and at most of finite type $\pi$. We refer the definitions and details to Appendix.
Therefore, we deduce
\begin{eqnarray}\nonumber
\|\widehat{G}(\rho)\|_{L^{2}(-\infty,\infty)}&\leq&
\|\frac{\sin\rho\pi}{\rho^{2}}\|_{L^{2}(-\infty,\infty)}\|\mathcal{F}(\sum_{i=1}^{N}\chi_{[0,a_{i}]}\widehat{q}(a_{i}-t))\|_{L^{2}(-\infty,\infty)}\\&&\hspace{-5pt}+\|\frac{\sum_{i=1}^{N}\sin\rho a_{i}}{\rho^{2}}\|_{L^{2}(-\infty,\infty)}\|\mathcal{F}(\chi_{[0,\pi]}(t)\widehat{q}(\pi-t))\|_{L^{2}(-\infty,\infty)}.\label{2223}
\end{eqnarray}
Using Parseval's identity to~(\ref{2223}), we obtain
\begin{eqnarray}\nonumber
\hspace{-2pt}\|\widehat{G}(\rho)\|_{L^{2}(-\infty,\infty)}&\leq&
\frac{1}{\sqrt{2\pi}}\,\|\frac{\sin\rho\pi}{\rho^{2}}\|_{L^{2}(-\infty,\infty)}\sum_{i=1}^{N}\|\chi_{[0,a_{i}]}(t)\widehat{q}(a_{i}-t)\|_{L^{2}(-\pi,\pi)}\\&&\hspace{-15pt}+\frac{1}{\sqrt{2\pi}}\,\|\frac{\sum_{i=1}^{N}\sin\rho a_{i}}{\rho^{2}}\|_{L^{2}(-\infty,\infty)}\|\chi_{[0,\pi]}(t)\widehat{q}(\pi-t)\|_{L^{2}(-\pi,\pi)},
\end{eqnarray}
and
\begin{eqnarray}\nonumber
\|\widehat{G}(\rho)\|_{L^{2}(-\infty,\infty)}&\leq&
\frac{N}{\sqrt{2\pi}}\,\|\frac{\sin\rho\pi}{\rho^{2}}\|_{L^{2}(-\infty,\infty)}\|\{\hat{c}_{k}\}_{k\in\mathbb{Z}}\|_{l^{2}}\\&&+\frac{1}{\sqrt{2\pi}}\,\|\frac{\sum_{i=1}^{N}\sin\rho a_{i}}{\rho^{2}}\|_{L^{2}(-\infty,\infty)}\|\{\hat{c}_{k}\}_{k\in\mathbb{Z}}\|_{l^{2}},\label{2116}
\end{eqnarray}
in which the set $\{\hat{c}_{k}\}_{k\in\mathbb{Z}}$ are the Fourier coefficients of $\mathcal{F}(\widehat{q}\,)(\rho)$ and $\mathcal{F}(\widehat{q}\,)(k)=\hat{c}_{k}$.
For $\rho\in\mathbb{C}$, we note 
\begin{eqnarray*}
&&\sin{\rho\pi}=O(e^{|\Im\rho|\pi});\\
&&\hspace{-21pt}\sum_{i=1}^{N}\,\sin\rho a_{i}=O(e^{|\Im\rho|a_{N}}),
\end{eqnarray*}
which implies that $\sin{\rho\pi}$ is an entire function of type $\pi$ and $\sum_{i=1}^{N}\sin\rho a_{i}$ is of type $a_{N}$ that are non-trivial entire functions with solid set with zero density. We refer the details to Appendix.
\begin{proposition}[Cartwright]\label{L2.2}
The entire function  $\widehat{G}(\rho)$ is of finite type  and can be represented as
\begin{equation}\label{G}
\widehat{G}(\rho)=c\rho^{m}\lim_{R\rightarrow\infty}\prod_{|\alpha_{k}|\leq R}\Big(1-\frac{\rho}{\alpha_{k}}\Big),\,k\in\mathbb{Z},
\end{equation}
where $\alpha_{k}\in\mathcal{Z}(\widehat{G})$, $c$ is a constant, and $\rho^{2}\widehat{G}(\rho)=o(1)$ for $\Re\rho\rightarrow\pm\infty$.
\end{proposition}
\begin{proof}
We observe that $\widehat{G}(\rho)$ is bounded over the real axis with non-trivial $q\in L^{2}(0,\pi)$. As a result of Cartwright's theory stated as~(\ref{CC}) in Theorem \ref{C} in Appendix, we begin that
\begin{equation}
\widehat{G}(\rho)=c\,\rho^me^{i\kappa
\rho}\lim_{R\rightarrow\infty}\prod_{|\alpha_{k}|\leq R}\Big(1-\frac{\rho}{\alpha_{k}}\Big),\,\alpha_{k}\in\mathcal{Z}(\widehat{G}),\label{2.10}
\end{equation}
where $c,m,\kappa$ are constants and $\kappa$ is real, and $\alpha_{k}\in\mathcal{Z}(\widehat{G})$.
Now we consider $\rho\rightarrow\pm\infty$ on real axis, we deduce that $$\widehat{G}(\rho)\rightarrow0,$$
using Riemann-Lebesgue Lemma to~(\ref{2.8}), and then deduce that
 the real constant $\kappa=0$. Then, we obtain
$$\widehat{G}(\rho)=c\,\rho^{m}\lim_{R\rightarrow\infty}\prod_{|\alpha_{k}|\leq R}\Big(1-\frac{\rho}{\alpha_{k}}\Big),$$
which is bounded on $\Im\rho=0$. Hence, $m\leq0$.
 Thus, we deduce that
\begin{equation}\label{214}
\widehat{G}(\rho)=c\lim_{R\rightarrow\infty}\prod_{|\alpha_{k}|\leq R}\Big(1-\frac{\rho}{\alpha_{k}}\Big),\,\alpha_{k}\in\mathcal{Z}(\widehat{G}).
\end{equation}
To retrieve the constant $c$, we let $\rho\rightarrow0$ in equation~(\ref{2.8})\begin{eqnarray}\nonumber
\widehat{G}(\rho)&\vspace{-5pt}=\vspace{-5pt}&-\frac{\sin\rho\pi}{\rho}\Big\{\sum_{i=1}^{N}\int_{0}^{a_{i}}\frac{\sin{(\rho t)}}{\rho}\,\widehat{q}(a_{i}-t)dt\Big\}\\&&+\,\frac{\sum_{i=1}^{N}\sin\rho a_{i}}{\rho}\int_{0}^{\pi}\frac{\sin(\rho t)}{\rho} \,\widehat{q}(\pi-t)dt.\label{213}
\end{eqnarray}
Using the formula
\begin{equation}
\int_{0}^{\alpha}u(t)\sin\rho tdt=o(e^{|\Im\rho|\alpha}),\,\alpha\in(0,\infty),
\end{equation}
we deduce from~(\ref{213}) that
\begin{equation}\label{221}
\widehat{G}(\rho)=e^{|\Im\rho|(\pi+a_{N})}o(\frac{1}{\rho^{2}}),\,|\rho|\gg0,\,\rho\in\mathbb{C}.
\end{equation}
The proportion is thus proven.

\end{proof}

\section{Results}
Let us denote by $L^{p}_{\sigma}$, $1\leq p<\infty$, the space of entire functions of exponential type $\leq\sigma$ that belong to the space $L^{p}(-\infty,\infty)$. Let us apply the Plancherel-P\'{o}lya theorem, and refer to \cite[p.\,149]{Levin2} as in Appendix. 
\begin{theorem}\label{T3.2}
Let $q^{\,j}\in L^{2}(0,\pi)$ be the potential parametrizing the system~(\ref{1.1}) with Fourier coefficients $\{c^{\,j}_{k}\}=\mathcal{F}(q^{\,j})(k)$, $j=1,2$, $k\in\mathbb{Z}$, and we write $$\widehat{q}(x)=q^{1}(x)-q^{2}(x).$$
Then, for some $C>0$ and some $h\gg0$, one has
\begin{equation}\nonumber
(h+1)(a_{N}+|\mbox{c.h. supp }\widehat{q}(t)|)\geq -\ln\frac{|C|\|\{\hat{c}_{k}\}_{k\in\mathbb{Z}}\|_{l^{2}}^{2}}{|\widehat{G}(x\pm ih)|^{2}},
\end{equation}
in which $|\mbox{c.h. supp }\widehat{q}(t)|$ means the measure of the convex hull of the effective support of $\widehat{q}(t)$.
\end{theorem}
\begin{proof}
Let us recall equation~(\ref{2116}) which is rephrased to
\begin{eqnarray}\nonumber
\|\widehat{G}(\rho)\|_{L^{2}(-\infty,\infty)}
&\leq&\{\frac{N}{2\sqrt{\pi}}\,\|\frac{\sin\rho\pi}{\rho^{2}}\|_{L^{2}(-\infty,\infty)}+\frac{1}{2\sqrt{\pi}}\,\|\frac{\sum_{i=1}^{N}\sin\rho a_{i}}{\rho^{2}}\|_{L^{2}(-\infty,\infty)}\}\\&&\times\,\|\{\hat{c}_{k}\}_{k\in\mathbb{Z}}\|_{l^{2}}\nonumber\\
&\leq&C_{\Im\rho}\,\|\{\hat{c}_{k}\}_{k\in\mathbb{Z}}\|_{l^{2}},\label{3.1}
\end{eqnarray}
for some constant $C_{\Im\rho}>0$. In this paper, we assume $\|\{\hat{c}_{k}\}_{k\in\mathbb{Z}}\|_{l^{2}}$ is non-zero and small. Otherwise, there is nothing to show. 
\par

Let us apply inequality~(\ref{(4.9)}) in Appendix to $\widehat{G}(\rho)$ with $p=2$. That is,
\begin{equation}\label{3.2}
|\widehat{G}(x+iy)|^{2}\leq\frac{2}{\pi}e^{\sigma(|y|+1)}\|\widehat{G}(\rho)\|_{L^{2}(-\infty,\infty)}^{2},
\end{equation}
in which we have shown $\widehat{G}(\rho)$ is an entire function of at most type $2\pi$ in~(\ref{221}).  Now we let $|\Im\rho|=|y|= h,$ for some $h\gg0$ in~(\ref{3.2}), and combine with~(\ref{3.1}). Hence,
\begin{equation}
|\widehat{G}(x\pm ih)|^{2}\leq Ce^{\sigma\,(h+1)}\|\{\hat{c}_{k}\}^{}_{k\in\mathbb{Z}}\|_{l^{2}}^{2},
\end{equation}
for some constant $C>0$. Referencing to~(\ref{221})  and Lemma \ref{Titchmarsh}, the type of $\widehat{G}(\rho)$ is at most $$a_{N}+|\mbox{c.h. supp } \widehat{q}(t)|,$$ in which we mean the convex hull of the effective support of $\widehat{q}=q^{1}-q^{2}\not\equiv0$. Using Definition \ref{D4.7} and equation~(\ref{(4.6)}), we deduce that
\begin{equation}
e^{-a_{N}-|\mbox{c.h. supp } \widehat{q}(t)|(h+1)}|\widehat{G}(x\pm ih)|^{2}\leq C\|\{\hat{c}_{k}\}_{k\in\mathbb{Z}}\|_{l^{2}}^{2}.
\end{equation}
Let us take natural logarithm on both sides, so we deduce 
\begin{equation}\nonumber
(h+1)(a_{N}+|\mbox{c.h. supp }\widehat{q}(t)|)\geq -\ln\frac{|C|\|\{\hat{c}_{k}\}_{k\in\mathbb{Z}}\|_{l^{2}}^{2}}{|\widehat{G}(x\pm ih)|^{2}}.
\end{equation}
Due to Lemma \ref{Levin163}, $\widehat{G}(x+iy)$ has no zero for $|y|= h\gg0$.
Hence, one proved the instability when the Fourier coefficients $\{\hat{c}^{j}_{k}\}_{k\in\mathbb{Z}}$, $j=1,2$, are approaching to each other.

\end{proof}
Alternatively, one may state the following instability result when one uses the other spectral data. In applicational sciences, zeros of Fourier transform may play a better role than Fourier coefficients which is called zero-crossing problem in signal analysis.
\begin{corollary}
Let $q^{\,j}\in L^{2}(0,\pi)$ be potential parametrizing the system~(\ref{1.1}) with Fourier transform $\mathcal{F}(q^{\,j})(k)$, $j=1,2$, with zero set $\{\beta^{j}_{k}\}_{k\in\mathbb{Z}}$. Then,
\begin{equation}\nonumber
(h+1)(a_{N}+|\mbox{c.h. supp }\widehat{q}(t)|)\geq -\ln\frac{|C|\|\{\beta^{1}_{k}-\beta^{2}_{k}\}_{k\in\mathbb{Z}}\|_{l^{2}}^{2}}{|\widehat{G}(x\pm ih|^{2}},\mbox{ for some constant }C.
\end{equation}
\end{corollary}
\begin{proof}
We use Theorem \ref{T164} in Appendix, and obtain
$$\{c^{j}_{k}\}_{k\in\mathbb{Z}}\mapsto f^{j}(z)=\sum_{k\in\mathbb{Z}}c^{j}_{k}\frac{\mathcal{F}(q^{\,j})(z)}{\mathcal{F}'(q^{\,j})(z)(z-\beta^{j}_{k})},\,j=1,2.$$
The inverse mapping is defined by the relation
$$f^{j}\mapsto\{f^{j}(\beta^{j}_{k})\}_{k\in\mathbb{Z}},\,j=1,2,$$
which is continuous.
Therefore, we deduce that
\begin{equation}\label{3111}
\sum_{k\in\mathbb{Z}}|f^{1}(\beta^{1}_{k})-f^{2}(\beta^{2}_{k})|^{2}\leq C\|f^{1}-f^{2}\|^{2}_{L^{2}(-\infty,\infty)}\leq C'\|\{c^{1}_{k}-c^{2}_{k}\}\|^{2}_{\,l^{2}},
\end{equation}
for non-zero constants $C$ and $C'$. Let us apply inequality~(\ref{(4.9)}) in Appendix again to $$f^{1}(x+iy)-f^{2}(x+iy),$$ and combine with~(\ref{3111}). Thus,
\begin{equation}\label{3.12}
|f^{1}(x+iy)-f^{2}(x+iy)|^{2}\leq\frac{2}{\pi}e^{\sigma(|y|+1)}\|f^{1}-f^{2}\|_{L^{2}(-\infty,\infty)}^{2}\leq C''e^{\sigma(|y|+1)}\|\{c^{1}_{k}-c^{2}_{k}\}\|^{2}_{\,l^{2}},
\end{equation}
in which we use~(\ref{3111}), $C''$ is some non-zero constant. The type of the function $\sigma$ for our case is at most $$a_{N}+|\mbox{c.h. supp } \widehat{q}\,|.$$
Then,
\begin{equation}\label{3.13}
|f^{1}(\Re\beta^{2}_{k}+i\Im\beta^{2}_{k})-f^{2}(\Re\beta^{2}_{k}+i\Im\beta^{2}_{k})|^{2}\leq C''e^{\sigma(|\Im\beta^{2}_{k}|+1)}\|\{c^{1}_{k}-c^{2}_{k}\}\|^{2}_{\,l^{2}}.
\end{equation}
Furthermore, by the triangular inequality,
\begin{equation}\label{3388}
\sum_{k\in\mathbb{Z}}|f^{1}(\beta^{1}_{k})-f^{2}(\beta^{2}_{k})|^{2}\geq\sum_{k\in\mathbb{Z}}\big||f^{1}(\beta^{1}_{k})-f^{1}(\beta^{2}_{k})|-|f^{1}(\beta^{2}_{k})-f^{2}(\beta^{2}_{k})|\big|^{2},
\end{equation}
in which $\sum_{k\in\mathbb{Z}}|f^{1}(\beta^{1}_{k})-f^{2}(\beta^{2}_{k})|$ is small for small $\|\{c^{1}_{k}-c^{2}_{k}\}\|^{2}_{\,l^{2}}$ by~(\ref{3111}). Choosing a suitable path from $\beta^{1}_{k}$ to $\beta^{2}_{k}$, we apply the mean value theorem to obtain
\begin{equation}\label{3315}
|f^{1}(\beta^{1}_{k})-f^{1}(\beta^{2}_{k})|=|(f^{1})'(d_{k})||\beta^{1}_{k}-\beta^{2}_{k}|,\mbox{ for some }d_{k},
\end{equation}
which leads to
\begin{equation}
\sum_{k\in\mathbb{Z}}|f^{1}(\beta^{1}_{k})-f^{2}(\beta^{2}_{k})|^{2}\geq\sum_{k\in\mathbb{Z}}\big||(f^{1})'(d_{k})||\beta^{1}_{k}-\beta^{2}_{k}|-|f^{1}(\beta^{2}_{k})-f^{2}(\beta^{2}_{k})|\big|^{2}.
\end{equation}
Inside which, we apply~(\ref{3111}) and~(\ref{3388}) to deduce
\begin{equation}\label{3.17}
C'\|\{c^{1}_{k}-c^{2}_{k}\}\|^{\,2}_{\,l^{2}}\geq \sum_{k\in\mathbb{Z}}\big||(f^{1})'(d_{k})||\beta^{1}_{k}-\beta^{2}_{k}|-C_{3}\|\{c^{1}_{k}-c^{2}_{k}\}\|_{\,l^{2}}\big|^{2},
\end{equation}\nonumber
in which $C_{3}^{2}=C''e^{\sigma(|\Im\beta^{2}_{k}|+1)}$, and we note that $\Im\beta^{2}_{k}$ is bounded for all $k$. Thus, $C_{3}^{2}$ is bounded uniformly. 
Moreover,
\begin{eqnarray}\nonumber
C'\|\{c^{1}_{k}-c^{2}_{k}\}\|^{\,2}_{\,l^{2}}&\geq&\sum_{k\in\mathbb{Z}}|(f^{1})'(d_{k})|^{2}|\beta^{1}_{k}-\beta^{2}_{k}|^{2}+C_{3}^{2}\|\{c^{1}_{k}-c^{2}_{k}\}\|_{\,l^{2}}^{\,2}\\&&-2C_{3}|(f^{1})'(d_{k})|\beta^{1}_{k}-\beta^{2}_{k}|\|\{c^{1}_{k}-c^{2}_{k}\}\|_{\,l^{2}}\nonumber\\\vspace{3pts}
\hspace{-2pt}&\hspace{-2pt}\geq\hspace{-2pt}&\hspace{-2pt}\sum_{k\in\mathbb{Z}}|(f^{1})'(d_{k})|^{2}|\beta^{1}_{k}-\beta^{2}_{k}|^{2}-2C_{3}\|\{c^{1}_{k}-c^{2}_{k}\}\|_{\,l^{2}}\hspace{-2pt}\sum_{k\in\mathbb{Z}}|(f^{1})'(d_{k})||\beta^{1}_{k}-\beta^{2}_{k}|,\nonumber
\end{eqnarray}
that is,
\begin{eqnarray}\nonumber
\|\{c^{1}_{k}-c^{2}_{k}\}\|_{\,l^{2}}&\geq&\frac{\sum_{k\in\mathbb{Z}}|(f^{1})'(d_{k})|^{2}|\beta^{1}_{k}-\beta^{2}_{k}|^{2}}{C'+2C_{3}\sum_{k\in\mathbb{Z}}|(f^{1})'(d_{k})|\,|\beta^{1}_{k}-\beta^{2}_{k}|}\\
&\geq&\frac{C_{4}\sum_{k\in\mathbb{Z}}|\beta^{1}_{k}-\beta^{2}_{k}|^{2}}{C'+2C_{3}\sum_{k\in\mathbb{Z}}|(f^{1})'(d_{k})|\,|\beta^{1}_{k}-\beta^{2}_{k}|},\nonumber
\end{eqnarray}
in which $\sum_{k\in\mathbb{Z}}|(f^{1})'(d_{k})||\beta^{1}_{k}-\beta^{2}_{k}|$ in the denominator converges due to~(\ref{3111}) and~(\ref{3315}). This combines with Theorem \ref{T3.2}, and corollary is thus proven.

\end{proof}

\section{Appendix: Entire Functions}

For reader's convenience,  we include a few mathematical linguistics.
\begin{definition}
Let $F(z)$ be an entire function. Let
\begin{equation}\nonumber
M_F(r):=\max_{|z|=r}|F(z)|.
\end{equation}
An entire function of $F(z)$ is said
to be a function of finite order if there exists a positive
constant $k$ such that the inequality
\begin{equation}\nonumber
M_F(r)<e^{r^k}
\end{equation}
is valid for all sufficiently large values of $r$. The greatest
lower bound of such numbers $k$ is called the order of the entire
function $F(z)$. By the type $\sigma$ of an entire function $F(z)$
of order $\rho$, we mean the greatest lower bound of positive
number $A$ for which asymptotically we have
\begin{equation}\nonumber
M_F(r)<e^{Ar^\rho}.
\end{equation}
That is,
\begin{equation}\nonumber
\sigma_{F}:=\limsup_{r\rightarrow\infty}\frac{\ln M_F(r)}{r^\rho}.
\end{equation}  If $0<\sigma_{F}<\infty$, then we say
$F(z)$ is of normal type or mean type. For $\sigma_{F}=0$, we say $F(z)$ is of minimal type.
\end{definition}
\begin{definition}
If an entire function $F(z)$ is of order one and of normal type,
then we say it is an entire function of exponential type (EFET).
\end{definition}
\begin{definition}\label{D3}
Let $F(z)$ be an analytic function of finite order $\rho$ in the 
angle $[\theta_1,\theta_2]$. We call the following quantity as the
indicator function of function $F(z)$.
\begin{equation}\nonumber
h_F(\theta):=\limsup_{r\rightarrow\infty}\frac{\ln|F(re^{i\theta})|}{r^{\rho}},
\,\theta_1\leq\theta\leq\theta_2.
\end{equation}
\end{definition}
The type of a function is connected to the maximal value of indicator function.
\begin{lemma}[Levin \cite{Levin},\,p.72]\label{L4}
The maximal value of indicator function $h_F(\theta)$ of
$F(z)$ on the interval $\alpha\leq\theta\leq\beta$ is equal to the
type $\sigma_F$ of this function inside the angle $\alpha\leq\arg
z\leq\beta$.
\end{lemma}
\begin{definition}\label{255}
Let $f(z)$ be an integral function of order $1$, and let
$N(f,\alpha,\beta,r)$ denote the number of the zeros of $f(z)$
inside the angle $[\alpha,\beta]$ and $|z|\leq r$. We define the
density function as
\begin{equation}\nonumber
\Delta_f(\alpha,\beta):=\lim_{r\rightarrow\infty}\frac{N(f,\alpha,\beta,r)}{r},
\end{equation}
and
\begin{equation}\nonumber
\Delta_f(\beta):=\Delta_f(\alpha_0,\beta),
\end{equation}
with some fixed $\alpha_0\notin E$ such that $E$ is at most a
countable set \cite{Levin,Levin2}. In particular, we denote the density function of $f$ on the open right/left half complex plane as $\Delta^{+}_{f}$/$\Delta^{-}_{f}$ respectively. Similarly, we can define the set density of a zero set $S$. Let $N(
S,r)$ be the number of the discrete elements of $S$ in $\{|z|<r\}$. We define
\begin{equation}\label{SR}
\Delta_S:=\lim_{r\rightarrow\infty}\frac{N(S,r)}{r}.
\end{equation}

\end{definition}
\begin{lemma}\label{L4.6}
Let $f$, $g$ be two entire functions. Then the following two
inequalities hold.
\begin{eqnarray}
&&h_{fg}(\theta)\leq h_{f}(\theta)+h_g(\theta),\mbox{ if one limit exists};\label{119}\\\label{120}
&&h_{f+g}(\theta)\leq\max_\theta\{h_f(\theta),h_g(\theta)\},
\end{eqnarray}
where the equality in~(\ref{119}) holds if one of the functions is of completely regular growth, and secondly, the equality~(\ref{120}) holds if the indicator of the two summands are not equal at some $\theta_0$.
\end{lemma}
\begin{proof}
 We can find
the details in \cite[p.\,51,\,p.\,52]{Levin}, and the sharpened results are discussed on page 159 and 160.

\end{proof}
\begin{definition}\label{D4.7}
The following quantity is called the width of indicator
diagram of entire function $f$:
\begin{equation}\nonumber
d=h_f(\frac{\pi}{2})+h_f(-\frac{\pi}{2}).
\end{equation}
\end{definition}
\begin{theorem}[Cartwright]\label{C}
Let $f$ be an entire function of exponential type with zero set $\{a_{k}\}$. We assume $f$ satisfies one of the
following conditions:
\begin{equation}\nonumber
\mbox{ the integral
}\int_{-\infty}^\infty\frac{\ln^+|f(x)|}{1+x^2}dx\mbox{ exists}.
\end{equation}
\begin{equation}\nonumber
|f(x)|\mbox{ is bounded on the real axis}.
\end{equation}
Then
\begin{enumerate}
\item all of the zeros of the function $f(z)$, except possibly
those of a set of zero density, lie inside arbitrarily small
angles $|\arg z|<\epsilon$ and $|\arg z-\pi|<\epsilon$, where the
density
\begin{equation}\nonumber
\Delta_f(-\epsilon,\epsilon)=\Delta_f(\pi-\epsilon,\pi+\epsilon)=\lim_{r\rightarrow\infty}
\frac{N(f,-\epsilon,\epsilon,r)}{r}
=\lim_{r\rightarrow\infty}\frac{N(f,\pi-\epsilon,\pi+\epsilon,r)}{r},
\end{equation}
is equal to $\frac{d}{2\pi}$, where $d$ is the width of the
indicator diagram in Definition \ref{D4.7}. Furthermore, the limit
$\delta=\lim_{r\rightarrow\infty}\delta(r)$ exists, where
$$
\delta(r):=\sum_{\{|a_k|<r\}}\frac{1}{a_k};
$$
\item moreover,
\begin{equation}\nonumber
\Delta_f(\epsilon,\pi-\epsilon)=\Delta_f(\pi+\epsilon,-\epsilon)=0;
\end{equation}
\item the function $f(z)$ can be represented in the form
\begin{equation}\label{CC}
f(z)=cz^me^{i\kappa
z}\lim_{r\rightarrow\infty}\prod_{\{|a_k|<r\}}(1-\frac{z}{a_k}),
\end{equation}
where $c,m,\kappa$ are constants and $\kappa$ is real;
\item the indicator
function of $f$ after certain renormalization can be reduced to be the form
\begin{equation}
h_f(\theta)=\sigma|\sin\theta|.\label{(4.6)}
\end{equation}
\end{enumerate}
\end{theorem}
\begin{proof}
We refer the last statement to Levin \cite[p.\,126]{Levin}.

\end{proof}
For reader's convenience, we include the following classical lemma proved by E. C. Titchmarsh \cite[Theorem\,IV]{Titchmarsh} which deals with the zero set of a Fourier transform. A modern version for functions in distributional sense can be found in \cite[Lemma\,1.3]{Tang}.
\begin{definition}
We say that $a$ and $b$ are the effective lower and upper limits of the integral if there is no number $\alpha>a$ such that
$$\int_{a}^{\alpha}|f(t)|dt=0,$$
and no number $\beta<0$ such that
$$\int_{\beta}^{b}|f(t)|dt=0.$$
\end{definition}
The detailed discussion is available in 
\cite[p.\,284]{Titchmarsh}.
\begin{lemma}[Titchmarsh]\label{Titchmarsh}
Let $u\in\mathcal{E}'(\mathbb{R})$, the space of distributions with compact support, then
\begin{equation}N_{\mathcal{F}(u)}(r)= \frac{|\mbox{c.h. supp } u|}{\pi}(r + o(1)),\label{T1}
\end{equation} in which
\begin{equation}
N_{f}(r)=\sum_{|z|\leq r}\frac{1}{2\pi}\oint_{z}\frac{f'(\omega)}{f(\omega)}d\omega,\,z\in\mathbb{C},\label{T2}
\end{equation}
and the phrase $|\mbox{c.h. supp }u|$ means  the convex hull of the effective support of $u$. Moreover, $N_{f}(r)$ is the counting function of the zeros of $f$ inside a ball of radius $r$ in $\mathbb{C}$, and we count the zeros according to their multiplicities.
\end{lemma}
\begin{proof}
We refer all details to \cite{Levin,Tang,Titchmarsh}.

\end{proof}
Let us denote by $L^{p}_{\sigma}$, $1<p<\infty$, the space of al entire functions of exponential type $\leq\sigma$ that belong to the space $L^{p}(-\infty,\infty)$. \begin{theorem}[Plancherel-P\'{o}lya]
Let us denote by $L^{p}_{\sigma}$, $1\leq p<\infty$, the space of al entire functions of exponential type $\leq\sigma$ that belong to the space $L^{p}(-\infty,\infty)$. Then, the following inequality holds.
\begin{equation}\label{(4.9)}
|f(x+iy)|^{p}\leq\frac{2}{\pi}e^{\sigma(|y|+1)}\|f\|_{L^{p}(-\infty,\infty)}^{p}.
\end{equation}
\end{theorem}
\begin{proof}
We refer the Plancherel-P\'{o}lya theorem to \cite[p.\,149]{Levin2}. 

\end{proof}
\begin{definition} \label{D4.12}
An entire function $F(z)$ is called a sine-type function with the width of indicator diagram $2\sigma$ if
\begin{enumerate}
\item $h_{F}(\frac{\pi}{2})= h_{F}(-\frac{\pi}{2})=\sigma$;
\item All of zeros of $F(z)$ are simple and satisfies the separation condition, that is 
$$\inf_{k\neq n}|\lambda_{k}-\lambda_{n}|=2\delta>0.$$
\item $\sup_{k}{|\Im\lambda_{k}|}=H<\infty$.
\item $0<c<|F(x+ih)|<C<\infty,\,\infty<x<\infty$, for some constants $c$, $C$, and $h$.
\end{enumerate}
\end{definition}
\begin{lemma}\label{Levin163}
The zeros of a sine-type function lie in a horizontal strip. 
\end{lemma}
\begin{proof}
We refer the proof to Levin \cite[p.\,163]{Levin2}.
\end{proof}
\begin{theorem}\label{T164}
Let $F(z)$ be a sine-type function with indicator diagram of width $2\sigma$, and let $\{\beta_{k}\}_{{k\in\mathbb{Z}}}$ its zero set. Then the mapping
$$\{c_{k}\}_{k\in\mathbb{Z}}\mapsto f(z)=\sum_{k\in\mathbb{Z}}c_{k}\frac{F(z)}{F'(\beta_{k})(z-\beta_{k})}$$
is an isomorphism between $l^{p}$ and $L^{p}_{\sigma}$ each $p\in(1,\infty)$. The series on the right-hand side converges in the $L^{p}(-\infty,\infty)$-norm. The inverse mapping is defined by the relation
$$f\mapsto\{f(\beta_{k})\}_{k\in\mathbb{Z}},$$
which is a continuous mapping.
\end{theorem}
\begin{proof}
We refer Levin \cite[p.\,165]{Levin2} for detailed introduction.

\end{proof}
\begin{definition}\label{D4.1}
We denote the following quantity as the zero density of an entire function $f$ of finite type:
\begin{equation}
\delta(f):=\lim_{r\rightarrow\infty}\frac{N_{f}(r)}{r}.
\end{equation}
\end{definition}
\begin{lemma}\label{L4.12}
Let $f$, $g$ be two entire functions of finite type with density functions $\delta(f)$ and
$\delta(g)$ respectively. Then, 
\begin{eqnarray*}
&&\delta(fg)=\delta(f)+\delta(g); \\
&&\delta(f+g)=\max\{\delta(f),\delta(g)\},\end{eqnarray*}
if the indicator diagrams of two functions are not equal.
\end{lemma}
\begin{proof}
We refer to B. Levin's book \cite[p.\,52]{Levin} for more detailed discussion and Lemma \ref{L4.6}. 
\end{proof}
\begin{lemma}\label{L2.4}
For $f(x)$ in $L^{2}(0,\pi)$, $\rho\in\mathbb{C}$, and $$\int_{0}^{\pi}e^{-i\rho x}f(x)dx=\int_{0}^{\pi}\cos{\rho x}f(x)dx-i\int_{0}^{\pi}\sin{\rho x}f(x)dx,$$ the entire functions $$\int_{0}^{\pi}e^{-i\rho x}f(x)dx,\, \int_{0}^{\pi}\cos{\rho x}f(x)dx,\mbox{ and }\int_{0}^{\pi}\sin{\rho x}f(x)dx$$ have the same zero density in $\mathbb{C}$.
\end{lemma}
\begin{proof}
Let us simply apply Lemma \ref{L4.6}, \ref{Titchmarsh}, Lemma \ref{L4.12}, and Theorem \ref{C}. The point is to see that the zero densities of $$\int_{0}^{\pi}e^{-i\rho x}f(x)dx,\, \int_{0}^{\pi}\cos{\rho x}f(x)dx,\mbox{ and }\int_{0}^{\pi}\sin{\rho x}f(x)dx$$
is about the effective integral supports of these functions. There are comments available in \cite{Titchmarsh}.

\end{proof}

\begin{acknowledgement}

The authors thank Professor Sergey Buterin at Saratov State University for the introduction to frozen argument problem during his visiting stay at Taiwan. The authors also ingenuously appreciate the research funding supported by NSTC under the project number 113-2115-M-131 -001. The content of this manuscript does not necessarily reflect the position or the policy of the administration, and no official endorsement should be inferred, neither.
\end{acknowledgement}

\end{document}